\documentclass[conference]{IEEEtran}
\usepackage{amsmath,amssymb,amsfonts,amsthm}
\usepackage{mathtools}
\usepackage{booktabs}
\usepackage{bm}
\usepackage{algorithm}
\usepackage{algorithmic}
\usepackage{graphicx}
\usepackage{cite}
\usepackage{url}
\usepackage{tikz}
\usepackage{subfigure}
\usetikzlibrary{shapes, arrows.meta, positioning, calc, backgrounds, fit, decorations.pathmorphing}
\usepackage[left=0.64in, right=0.625in, top=0.76in, bottom=1.08in]{geometry}

\newcommand{\argmin}{\operatornamewithlimits{arg\,min}}
\newcommand{\argmax}{\operatornamewithlimits{arg\,max}}

\newcommand{\wrap}{\mathrm{wrap}}

\newcommand{\Ree}{\operatorname{Re}}

\newcommand{\jj}{\mathrm{j}} 

\theoremstyle{plain}
\newtheorem{theorem}{Theorem}

\theoremstyle{definition}

\theoremstyle{remark}

\title{LO-Free Joint Communication and Sensing via Inter-Antenna Cross-Correlation and Graph-Based Spatial Phase Inference}

\author{
     Hasan Atalay Günel\IEEEauthorrefmark{1}\IEEEauthorrefmark{3}, Mohaned Chraiti\IEEEauthorrefmark{2}, and Ali Görçin\IEEEauthorrefmark{1}\IEEEauthorrefmark{3} \\
    \IEEEauthorrefmark{1}Communications and Signal Processing Research (HİSAR) Lab., T{\"{U}}B{\.{I}}TAK B{\.{I}}LGEM, Kocaeli, Turkey \\
 \IEEEauthorrefmark{2} Department of Electronics Engineering, Sabanci University, Istanbul, Turkey\\ 
\IEEEauthorrefmark{3}
Department of Electrical and Electronic Engineering, Istanbul Technical University, Istanbul, Turkey.\\
    Emails:
    hasan.gunel@tubitak.gov.tr,  
    mohaned.chraiti@sabanciuniv.edu,
    and 
    ali.gorcin@tubitak.gov.tr.
}

\begin{document}
\maketitle

\begin{abstract}


Joint communication and sensing (JCAS) typically rely on coherent downconversion to recover the phase relationships required for array processing. Meanwhile, Local Oscillators (LOs) are a major source of cost, power consumption, and implementation complexity in millimeter-wave (mmWave) and sub-THz receivers. Existing LO-free receiver designs are typically based on envelope detection or related non-coherent operations that do not preserve inter-branch phase information, which limits their applicability to JCAS. This work proposes an LO-free JCAS receiver architecture that leverages pairwise inter-branch correlation processing to suppress the common carrier component and to synthesize relative-phase observables across the antenna array, enabling both data communication and Direction-of-Arrival (DoA) estimation. The transmitted symbols are designed to induce distinct phase-difference patterns, such that the resulting correlation phases contain both a data-dependent component and a DoA-dependent component. We formulate recovery as inference over a correlation graph, where branches are nodes and pairwise correlations are edges, and show that the resulting cycle-consistent redundancy enables robust relative-phase recovery under noise and perturbations. We further derive a topology-aware Cramér-Rao lower bound for DoA estimation under a locally unwrapped approximation. Numerical results confirm that increasing graph connectivity improves both bit-error rate and DoA accuracy, with sensing performance approaching the derived bound. 

\end{abstract}


\section{Introduction}
Joint communication and sensing (JCAS) is increasingly adopted, where the same radio platform supports both data transmission and spatial inference~\cite{he2024isac}. Most existing designs rely on coherent receiver architectures equipped with local oscillators (LOs), which enable frequency translation and recovery of amplitude, carrier phase, and inter-antenna phase relationships. These observables are central to coherent demodulation, beamforming, and sensing tasks such as direction-of-arrival (DoA) estimation.

At millimeter-wave (mmWave) and sub-terahertz (sub-THz) frequencies, LO-empowered coherent receivers are demanding in terms of power consumption, hardware complexity, implementation cost, and sensitivity to oscillator impairments such as phase noise~\cite{yang2019hard}. This has motivated extensive work on LO-relaxed and LO-free front-ends, particularly for low-power and Internet-of-Things receivers. Representative examples include envelope-detection, square-law, and self-mixing architectures, which avoid continuous coherent LO operation and instead rely on noncoherent observables such as signal magnitude, energy, or beat-frequency components~\cite{wentzloff2020review,wurx2022,jlpea2025wvrx,selfdemod2010}. While such architectures are attractive from a power-efficiency perspective, they do not naturally preserve the branch-to-branch relative phase structure required for coherent spatial processing. As a result, sensing functions such as DoA estimation and beamspace inference cannot be directly supported in their standard form~\cite{godara97,rajabi2018irr}.


The feasibility of JCAS under LO-free reception remains largely open. In particular, it remains unclear how to construct a measurement domain that simultaneously supports data detection and DoA inference when conventional coherent phase recovery is unavailable. The key observation is that pairwise inter-branch correlations can suppress the common carrier-dependent oscillatory term while retaining relative phase information across the array. Building on this principle, we develop an LO-free JCAS receiver in which transmitted signaling induces distinguishable correlation-phase patterns, while the received pairwise phase differences also retain the DoA-dependent spatial signature. To improve robustness, we represent the resulting relative-phase measurements on a graph, where branches are nodes and pairwise correlations are edges. This graph-based formulation enables topology-aware processing for joint data detection and DoA estimation directly from noisy correlation-domain observables. We characterize the resulting communication and sensing performance, including a topology-aware Cram\'er-Rao lower bound (CRLB) for DoA estimation and a topology-dependent noise-reduction result. The main contributions of this paper are summarized as follows:

\begin{itemize}
\item We propose an LO-free multi-branch correlation JCAS technique in which both communication and DoA estimation functionalities are simultaneously supported.
\item We develop a graph-structured framework, termed Graph-based Spatial Manifold Communications (G-SPMC), for joint data detection and DoA estimation from noisy pairwise phase-difference observations.
\item We characterize the sensing limits through a topology-aware CRLB for DoA estimation, and establish uniqueness of the reconstructed phase-difference observable, which guarantees unambiguous recovery of both the transmitted data and the DoA.
\end{itemize}


\section{Spatial Phase Manifold Communications}

\subsection{Transmission and received signal model}

We consider a narrowband link with an $N_t$-antenna transmitter and an LO-free receiver with $N_r$ receive antennas (also called ports). The receiver aims to detect the transmitted data while simultaneously estimating the transmitter's DoA. In the $k$th symbol interval, the input bits are mapped to one codeword from a finite spatial constellation
\begin{equation}
\mathcal S \triangleq \{\bm s^{(1)},\dots,\bm s^{(M)}\},
\qquad \bm s^{(m)}\in\mathbb C^{N_t},
\end{equation}
where each codeword represents $\log_2 M$ information bits. The transmitted spatial codeword at interval $k$ is denoted by
\begin{equation}
\bm s_k \triangleq \big[e^{\jj\varphi_1(k)},\dots,e^{\jj\varphi_{N_t}(k)}\big]^{\mathsf T}\in\mathcal S,
\end{equation}
yielding the transmitted signal $\bm x_{\mathrm{tx}}(t)=x(t)\bm s_k,$ where $x(t)$ denotes the common transmit pulse. Under a dominant plane wave with direction $\theta$, the narrowband channel is modeled as
\begin{equation}
\bm H(\theta)=\bm D_r(\theta)\,\widetilde{\bm H},
\label{eq:H_factorized}
\end{equation}
where $\widetilde{\bm H}\in\mathbb C^{N_r\times N_t}$ denotes the direction-independent channel factor, and
\begin{equation}
\bm D_r(\theta)\triangleq \mathrm{diag}\!\left(e^{\jj\gamma_1(\theta)},\dots,e^{\jj\gamma_{N_r}(\theta)}\right)
\end{equation}
collects the receive-side geometric phases induced by the arrival direction $\theta$. Specifically,
\begin{equation}
\gamma_i(\theta)=-\frac{2\pi}{\lambda}\bm p_i^{\mathsf T}\bm u(\theta),
\end{equation}
where $\bm p_i$ denotes the position of the $i$th receive port, $\lambda$ is the carrier wavelength, and $\bm u(\theta)$ is the receiver's array response vector. Define
\begin{equation}
\widetilde{\bm g}_k \triangleq \widetilde{\bm H}\bm s_k,
\qquad
\bm g_k(\theta)\triangleq \bm H(\theta)\bm s_k=\bm D_r(\theta)\widetilde{\bm g}_k.
\end{equation}
Writing the $i$th entry of $\widetilde{\bm g}_k$ as
\begin{equation}
[\widetilde{\bm g}_k]_i=\tilde\alpha_i(k)e^{\jj\tilde\beta_i(k)},
\qquad \tilde\alpha_i(k)\in\mathbb R_+,
\end{equation}
we obtain
\begin{equation}
[\bm g_k(\theta)]_i
=
\tilde\alpha_i(k)e^{\jj(\gamma_i(\theta)+\tilde\beta_i(k))}.
\end{equation}
Accordingly, define the effective phase at receive port $i$ as
\begin{equation}
\beta_i(\theta,k)\triangleq \gamma_i(\theta)+\tilde\beta_i(k).
\label{eq:effective_phase_def}
\end{equation}
The analytic received RF signal at port $i$ is then modeled as
\begin{equation}
y_i(t)=[\bm g_k(\theta)]_i\,x(t)\,e^{\jj\psi_c(t)}+n_i(t),
\label{eq:yi_corrected}
\end{equation}
where $\psi_c(t)\triangleq 2\pi f_c t+\phi_c(t)$ is the common carrier phase process observed across receive ports, and $n_i(t)$ denotes additive receiver noise. The term $\phi_c(t)$ models the common phase fluctuation of the incident carrier or any residual phase noise. The inter-port phase difference is therefore
\begin{equation}
\begin{aligned}
\Delta\beta_{ij}(\theta,k)
&\triangleq
\beta_i(\theta,k)-\beta_j(\theta,k)
\\
&=
\Delta\beta_{ij}^{\mathrm{geo}}(\theta)+\Delta\beta_{ij}^{\mathrm{sym}}(k),
\label{eq:phase_diff_decomp}
\end{aligned}
\end{equation}
where
\begin{equation}
\Delta\beta_{ij}^{\mathrm{geo}}(\theta)
=
\gamma_i(\theta)-\gamma_j(\theta)
=
-\frac{2\pi}{\lambda}\bm b_{ij}^{\mathsf T}\bm u(\theta),
\qquad
\bm b_{ij}\triangleq \bm p_i-\bm p_j,
\label{eq:geo_phase_diff_corrected}
\end{equation}
and
\begin{equation}
\Delta\beta_{ij}^{\mathrm{sym}}(k)
\triangleq
\tilde\beta_i(k)-\tilde\beta_j(k).
\end{equation}

\noindent Thus, the received pairwise phase pattern is jointly determined by the receive-array geometry and the transmit-induced effective phase profile. The signaling codebook $\{\bm s_k\}$ is designed such that the resulting symbol-dependent phase-difference set $\{\Delta\beta_{ij}^{\mathrm{sym}}(k)\}_{i,j}$ is distinguishable across symbols, while the geometric term $\Delta\beta_{ij}^{\mathrm{geo}}(\theta)$ remains available for DoA inference.

\begin{figure}[t]
    \centering
    \includegraphics[width=1\linewidth]{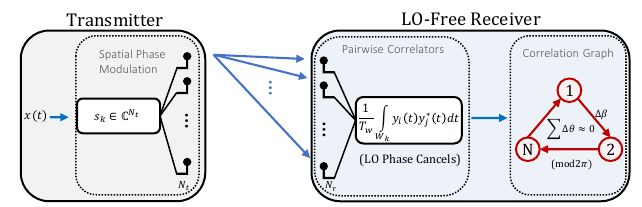}
    \caption{System model of the proposed LO-free G-SPMC receiver} 
    \label{fig:system}
    \vspace{-0.5cm}
\end{figure}


\subsection{SPMC receiver}



The receiver is assumed to be LO-free. It extracts both the transmitted information and the DoA from pairwise inter-port correlation observables. In practice, such observables can be formed by pairwise mixing of the received branch signals, followed by low-pass filtering and windowed accumulation, after which digital processing is applied.

Consider the pairwise correlation matrix over the observation window $\mathcal W_k$:
\begin{equation}
\widehat{\bm C}_k \triangleq \frac{1}{T_w}\int_{\mathcal W_k}\bm y(t)\bm y(t)^{\mathsf H}\,dt,
\end{equation}
where $\bm y(t)\triangleq [y_1(t),\dots,y_{N_r}(t)]^{\mathsf T}$ and $T_w$ is the window length. Its $(i,j)$th entry is
\begin{equation}
\widehat c_{ij}(k)\triangleq \frac{1}{T_w}\int_{\mathcal W_k} y_i(t)y_j^*(t)\,dt.
\end{equation}
Substituting \eqref{eq:yi_corrected} gives
\begin{equation}
y_i(t)y_j^*(t)
=
[\bm g_k(\theta)]_i[\bm g_k(\theta)]_j^*|x(t)|^2
+
\text{cross terms noise},
\end{equation}
where the common carrier factor $e^{\jj\psi_c(t)}$ cancels exactly in the product $y_i(t)y_j^*(t)$. Hence, for $i\neq j$,
\begin{equation}
\widehat c_{ij}(k)
=
\kappa_{ij}(k)e^{\jj\Delta\beta_{ij}(\theta,k)}+v_{ij}(k),
\label{eq:cij_phasor_corrected}
\end{equation}
where
\begin{equation}
\kappa_{ij}(k)\triangleq
\tilde\alpha_i(k)\tilde\alpha_j(k)\,
\frac{1}{T_w}\int_{\mathcal W_k}|x(t)|^2dt
\in\mathbb R_+,
\end{equation}
and $v_{ij}(k)$ is the effective complex disturbance induced by receiver noise and finite-window averaging. 

Let $G=(V,E)$ be the correlation graph, where $V=\{1,\dots,N_r\}$ indexes the receive ports and each edge $(i,j)\in E$ corresponds to one retained pairwise correlation measurement. For each edge, we form the normalized correlation
\begin{equation}
\widehat\rho_{ij}(k)\triangleq
\frac{\widehat c_{ij}(k)}{\sqrt{\widehat c_{ii}(k)\widehat c_{jj}(k)}}.
\label{eq:rhoij_def}
\end{equation}
Since inference depends only on edge phase, we define
\begin{equation}
\widehat\psi_{ij}(k)\triangleq \angle \widehat\rho_{ij}(k)\in(-\pi,\pi],
\label{eq:phase_only_psi_local}
\end{equation}
and the associated phase-only phasor
\begin{equation}
\bar\rho_{ij}(k)\triangleq e^{\jj\widehat\psi_{ij}(k)}
=\frac{\widehat\rho_{ij}(k)}{|\widehat\rho_{ij}(k)|},
\qquad |\widehat\rho_{ij}(k)|>0.
\label{eq:phase_only_rho_local}
\end{equation}
In the noiseless case, $\widehat\psi_{ij}(k)$ reduces to the wrapped inter-port phase difference in $(-\pi,\pi]$
\begin{equation}
\bar\rho_{ij}(k)=e^{\jj\Delta\bar{\beta}_{ij}(\theta,k)},
\,\,
\Delta\bar{\beta}_{ij}(\theta,k)\triangleq \wrap\!\big(\Delta\beta_{ij}(\theta,k)\big).
\label{eq:rhoij}
\end{equation}
Using \eqref{eq:phase_diff_decomp}, the edge phase is therefore modeled as
\begin{equation}
\Delta\bar{\beta}_{ij}(\theta,k)
=
\wrap\!\Big(
\Delta\beta^{\mathrm{geo}}_{ij}(\theta)
+
\Delta\beta^{\mathrm{sym}}_{ij}(k)
+
\nu_{ij}(k)
\Big),
\label{eq:edge_superposition}
\end{equation}
where $\nu_{ij}(k)$ denotes the induced wrapped phase perturbation. In the noiseless single-source setting, the edge phases are induced by pairwise differences of the underlying port phases, namely $
\Delta\beta_{ij}(\theta,k)$ defined in \eqref{eq:phase_diff_decomp}. Therefore, for any directed cycle $c$ in $G$, the edge-phase differences telescope:
\begin{equation}
\sum_{(i,j)\in c}\Delta\beta_{ij}(\theta,k)=0
\quad \mathrm{mod}\; 2\pi.
\label{eq:meas_cycle_consistency}
\end{equation}


\section{Graph-Structured SPMC}

Information is embedded in the phase pattern across the receive ports. The correlation operation yields a phase-difference measurement for each port pair, which can be interpreted as an edge phase on the correlation graph.

\subsection{Spatial Phase Graph Representation and Operators}

The receiver operates on pairwise phase differences. Let
\begin{equation}
\bm z_k \triangleq \big[e^{\jj\beta_1(\theta,k)},\dots,e^{\jj\beta_{N_r}(\theta,k)}\big]^{\mathsf T}\in (\mathbb S^1)^{N_r}
\label{eq:zk_def}
\end{equation}
denote the latent node-phase vector at the $k$th observation window, where $\beta_i(\theta,k)$ is the effective phase at receive port $i$ defined in \eqref{eq:effective_phase_def}. For any edge $(i,j)\in E$, we have
\begin{equation}
e^{\jj \Delta\beta_{ij}(\theta,k)} = z_{k,i} z_{k,j}^*.
\label{eq:edge_from_node_phases}
\end{equation}
This representation implies cycle consistency in the noiseless single-source setting. Specifically, for any directed cycle
$c=(i_0,i_1,\dots,i_L=i_0)$ in $G$,
\begin{equation}
\prod_{\ell=0}^{L-1} e^{\jj \Delta\beta_{i_\ell i_{\ell+1}}(\theta,k)} = 1
\,\Longleftrightarrow\,
\sum_{\ell=0}^{L-1} \Delta\beta_{i_\ell i_{\ell+1}}(\theta,k)=0
\ \mathrm{mod}\; 2\pi,
\label{eq:cycle_consistency_prod}
\end{equation}
as already indicated by \eqref{eq:meas_cycle_consistency}.

To attenuate quasi-static geometric contributions, we define the edge-phase increment operator
\begin{equation}
\delta_{ij}(k)\triangleq
\angle\!\Big(\bar{\rho}_{ij}(k)\bar{\rho}_{ij}^*(k-1)\Big),
\label{eq:edge_increment}
\end{equation}
where $\bar{\rho}_{ij}(k)$ is  defined in \eqref{eq:phase_only_rho_local}. Considering \eqref{eq:phase_diff_decomp}, we obtain
\begin{equation}
\delta_{ij}(k)
=
\wrap\!\big(\Delta\beta_{ij}(\theta,k)-\Delta\beta_{ij}(\theta,k-1)\big).
\label{eq:edge_increment_diff}
\end{equation}
If the geometric term $\Delta\beta_{ij}^{\mathrm{geo}}(\theta)$ varies negligibly over two consecutive windows, then
\begin{equation}
\delta_{ij}(k)
\approx
\wrap\!\Big(
\Delta\beta^{\mathrm{sym}}_{ij}(k)-\Delta\beta^{\mathrm{sym}}_{ij}(k-1)
+\Delta\nu_{ij}(k)
\Big),
\label{eq:edge_increment_geo_free}
\end{equation}
where $\Delta\nu_{ij}(k)$ denotes the corresponding inter-window noise. Thus, $\delta_{ij}(k)$ primarily captures symbol-to-symbol evolution.

\subsection{Manifold Inference via Angular Synchronization}

Given the phase-only edge measurements $\{\bar{\rho}_{ij}(k)\}_{(i,j)\in E}$, the objective is to estimate the latent node-phase vector
\(
\bm z_k\in(\mathbb S^1)^{N_r}.
\)
In the noiseless case, \eqref{eq:rhoij} and \eqref{eq:edge_from_node_phases} give
$\bar{\rho}_{ij}(k)=\bar{\rho}_{ij}(k)\approx z_{k,i}z_{k,j}^*=e^{\jj\Delta\bar{\beta}_{ij}(\theta,k)}.$
Considering the phase perturbations (noise), we adopt the unit-modulus multiplicative model
\begin{equation}
\bar{\rho}_{ij}(k)=z_{k,i}z_{k,j}^*\,\eta_{ij}(k),
\label{eq:rho_noise_model}
\end{equation}
where $\eta_{ij}(k)\in\mathbb S^1$ captures the wrapped edge-phase noise. Under a maximum-likelihood formulation with von Mises edge noise, estimation reduces to angular synchronization~\cite{fisher1993circular}:
\begin{equation}
\widehat{\bm z}_k
\in
\argmax_{\bm z\in(\mathbb S^1)^{N_r}}
\sum_{(i,j)\in E}
w_{ij}(k)\,
\Ree\!\big(\bar{\rho}_{ij}(k) z_i z_j^*\big),
\label{eq:synchronization}
\end{equation}
where $w_{ij}(k)\ge 0$ weights the reliability of edge $(i,j)$, for example according to the magnitude of the corresponding sample correlation.


To obtain an efficient initializer, we construct the Hermitian measurement matrix $\bm Y_k\in\mathbb C^{N_r\times N_r}$ as
\begin{equation}
[\bm Y_k]_{ij}\triangleq
\begin{cases}
w_{ij}(k)\,\bar{\rho}_{ij}(k), & (i,j)\in E,\\[2pt]
\overline{w_{ji}(k)\,\bar{\rho}_{ji}(k)}, & (j,i)\in E,\\[2pt]
0, & \text{otherwise},
\end{cases}
\qquad [\bm Y_k]_{ii}=0,
\label{eq:Yk_def}
\end{equation}
which ensures $\bm Y_k=\bm Y_k^{\mathsf H}$. Let $\bm v_k$ denote the leading eigenvector of $\bm Y_k$. The spectral initializer is then obtained by entrywise projection onto the unit circle:
\begin{equation}
\widehat z_{k,i}^{(0)}\triangleq \exp\!\big(\jj\,\angle(v_{k,i})\big).
\label{eq:spectral}
\end{equation}

\section{Theoretical Analysis}
\label{sec:theory_alg}
This section establishes the fundamental limits of G-SPMC. We study identifiability of graph-consistent phase states, derive a topology-aware CRLB for DoA estimation, and prove the denoising gain of projection onto the gradient subspace in the locally unwrapped regime.

\subsection{Identifiability Analysis}

The graph-based receiver relies only on relative phase information and does not directly observe absolute node phases. We therefore first establish identifiability of the latent node-phase vector from the pairwise edge-phase observations. In this setting, identifiability asks whether two different node-phase vectors can produce the same edge-phase measurements. If so, the phase state cannot be uniquely recovered.

The following theorem shows that, under the locally unwrapped model, the only inherent ambiguity is an additive constant on each connected component. In particular, for a connected receiver correlation graph, this reduces to the usual global phase-offset ambiguity. Hence, graph connectivity is the key structural condition for meaningful phase recovery.

\begin{theorem}
\label{thm:identifiability}
Let $G=(V,E)$ be a graph with $N_r$ nodes and oriented incidence matrix
$\mathbf B\in\{-1,0,1\}^{N_r\times |E|}$. Assume a locally unwrapped edge-phase representation
\begin{equation}
\bm\psi = \mathbf B^{\mathsf T}\bm\beta,
\label{eq:local_unwrapped_model}
\end{equation}
where $\bm\beta\in\mathbb R^{N_r}$ denotes the latent node-phase vector and $\bm\psi\in\mathbb R^{|E|}$ collects the corresponding unwrapped edge-phase differences. If $G$ has $C$ connected components with node sets $\{V_c\}_{c=1}^C$, then $\bm\beta$ is identifiable from $\bm\psi$ only up to additive constants on each connected component, i.e.,
\begin{equation}
\bm\beta' \text{ yields the same } \bm\psi
\quad\Longleftrightarrow\quad
\bm\beta'=\bm\beta+\sum_{c=1}^C a_c\,\mathbf 1_{V_c},
\end{equation}
for some scalars $\{a_c\}_{c=1}^C$. In particular, if $C=1$, then $\bm\beta$ is identifiable up to a single global constant $c\mathbf 1_{N_r}$.
\end{theorem}

\begin{proof}
Suppose $\bm\beta$ and $\bm\beta'$ are two node-phase vectors such that
\begin{equation}
\mathbf B^{\mathsf T}\bm\beta=\bm\psi
\qquad\text{and}\qquad
\mathbf B^{\mathsf T}\bm\beta'=\bm\psi.
\end{equation}
Subtracting the two equalities gives $\mathbf B^{\mathsf T}(\bm\beta-\bm\beta')=\mathbf 0$. Hence, two node-phase vectors are observationally indistinguishable if and only if their difference lies in the nullspace of $\mathbf B^{\mathsf T}$:
\begin{equation}
\bm\beta-\bm\beta' \in \mathcal N(\mathbf B^{\mathsf T}).
\end{equation}
Therefore, the identifiability question reduces to determining the structure of $\mathcal N(\mathbf B^{\mathsf T})$.

Let $V=\bigcup_{c=1}^C V_c$ denote the partition of the graph into its $C$ connected components, and let $\mathbf 1_{V_c}\in\mathbb R^{N_r}$ be the indicator vector of component $V_c$. For each $c$, every edge either lies entirely within $V_c$ or lies entirely outside it. Since each column of $\mathbf B$ contains one $+1$ entry and one $-1$ entry at the incident nodes of the corresponding edge, the contributions cancel for edges inside $V_c$ and vanish for edges outside $V_c$. Therefore,
\begin{equation}
\mathbf B^{\mathsf T}\mathbf 1_{V_c}=\mathbf 0,\qquad c=1,\dots,C,
\end{equation}
which implies $
\mathrm{span}\{\mathbf 1_{V_1},\dots,\mathbf 1_{V_C}\}\subseteq \mathcal N(\mathbf B^{\mathsf T}).$ Moreover, a standard result from algebraic graph theory gives $
\mathrm{rank}(\mathbf B)=N_r-C.$ Hence, by rank--nullity,
\begin{equation}
\dim \mathcal N(\mathbf B^{\mathsf T})
= N_r-\mathrm{rank}(\mathbf B^{\mathsf T})
= C.
\end{equation}
Since the vectors $\mathbf 1_{V_1},\dots,\mathbf 1_{V_C}$ are linearly independent, their span also has dimension $C$. Therefore,
\begin{equation}
\mathcal N(\mathbf B^{\mathsf T})
=
\mathrm{span}\{\mathbf 1_{V_1},\dots,\mathbf 1_{V_C}\}.
\end{equation}
It follows that if $\bm\beta'$ yields the same $\bm\psi$ as $\bm\beta$, then $
\bm\beta-\bm\beta' \in \mathcal N(\mathbf B^{\mathsf T}),$ and hence there exist scalars $a_1,\dots,a_C$ such that
\begin{equation}
\bm\beta-\bm\beta' = \sum_{c=1}^C a_c\,\mathbf 1_{V_c}.
\end{equation}
Equivalently, $
\bm\beta'=\bm\beta+\sum_{c=1}^C \tilde a_c\,\mathbf 1_{V_c},$ 
where $\tilde a_c=-a_c$. Thus, the latent node-phase vector is uniquely determined by $\bm\psi$ up to one additive constant on each connected component. If $C=1$, this reduces to the usual global phase-offset ambiguity
\begin{equation}
\bm\beta'=\bm\beta+c\mathbf 1_{N_r}.
\end{equation}
This proves the claim.
\end{proof}

Theorem~\ref{thm:identifiability} implies that the measured edge-phase differences can recover the phase pattern across the graph, but not the absolute phase at each port. In $
\bm\beta'=\bm\beta+\sum_{c=1}^C a_c\,\mathbf 1_{V_c},$ 
each $a_c$ is just a scalar constant added to all nodes in the $c$th connected component. Such a uniform shift does not change any phase difference within that component. Therefore, for a connected graph, all node phases are determined up to one common offset; for a disconnected graph, each disconnected part has its own separate offset.

\subsection{Estimation Bounds: Topology-Aware CRLB}
\label{sec:top-CRLB}

We quantify a precision limit for estimating the direction parameter $\theta$ from edge phases after conditioning the observations to depend on geometry only. Specifically, we assume either (i) sensing-only pilot windows are used, or (ii) the symbol-dependent phase term $\Delta\beta^{\mathrm{sym}}_{ij}(k)$ is removed using the detected codeword. Under a standard error (locally unwrapped) approximation, the resulting edge-phase measurements satisfy
\begin{equation}
\hat{\psi}_{ij}\sim \mathcal{N}\big(\psi_{ij}(\theta),\sigma_{ij}^2\big),
\qquad
\psi_{ij}(\theta)= -\frac{2\pi}{\lambda}\bm b_{ij}^{\mathsf T}\bm u(\theta),
\label{eq:geom_only_phase_model}
\end{equation}
where $\psi_{ij}(\theta)$ coincides with the geometric edge-phase difference $\Delta\beta^{\mathrm{geo}}_{ij}(\theta)$ defined in \eqref{eq:geo_phase_diff_corrected}. In practice, this geometry-only conditioning is implemented either by allocating sensing-only pilot windows or by subtracting the detected symbol-dependent phase contribution. Likewise, for communication, the quasi-static geometric term can be removed via pilot-based calibration or attenuated by the edge-increment operator in~\eqref{eq:edge_increment}.

\begin{theorem}[Topology-aware Fisher Information and CRLB]
\label{thm:crlb}
Under the conditioned, locally unwrapped edge-phase model in \eqref{eq:geom_only_phase_model}, the Fisher information for estimating $\theta$ from one snapshot on $G$ is
\begin{equation}
\mathcal{I}(\theta)
=\left(\frac{2\pi}{\lambda}\right)^2
\sum_{(i,j)\in E}\frac{1}{\sigma_{ij}^2}\,
\left(\bm b_{ij}^{\mathsf T}\dot{\bm u}(\theta)\right)^2,
\label{eq:fisher_info}
\end{equation}
where $\dot{\bm u}(\theta)\triangleq \frac{\partial \bm u(\theta)}{\partial \theta}$.
Consequently, any unbiased estimator satisfies
\begin{equation}
\mathrm{var}(\hat{\theta})\ge \mathcal{I}(\theta)^{-1}.
\end{equation}
\end{theorem}

\begin{proof}
Let $\hat{\bm\psi}\in\mathbb R^{|E|}$ collect the conditioned, locally unwrapped edge-phase measurements. Under the independent Gaussian model, the conditional log-likelihood is
\[
\mathcal L(\theta)
= \mathrm{const}
-\frac12\sum_{(i,j)\in E}\frac{1}{\sigma_{ij}^2}\big(\hat\psi_{ij}-\psi_{ij}(\theta)\big)^2 .
\]
Differentiating with respect to $\theta$ yields
\[
\frac{\partial \mathcal L}{\partial\theta}
=\sum_{(i,j)\in E}\frac{1}{\sigma_{ij}^2}\big(\hat\psi_{ij}-\psi_{ij}(\theta)\big)\,
\frac{\partial \psi_{ij}(\theta)}{\partial\theta},
\]
and
\[
\frac{\partial^2 \mathcal L}{\partial\theta^2}
=
\sum_{(i,j)\in E}\frac{1}{\sigma_{ij}^2}
\left[
\big(\hat\psi_{ij}-\psi_{ij}(\theta)\big)\frac{\partial^2 \psi_{ij}(\theta)}{\partial\theta^2}
-\left(\frac{\partial \psi_{ij}(\theta)}{\partial\theta}\right)^2
\right]
\]
Since $\mathbb E[\hat\psi_{ij}\mid\theta]=\psi_{ij}(\theta)$, the first term vanishes in expectation, and therefore
\[
\mathcal I(\theta)\triangleq -\mathbb E\!\left[\frac{\partial^2 \mathcal L}{\partial\theta^2}\right]
=\sum_{(i,j)\in E}\frac{1}{\sigma_{ij}^2}\left(\frac{\partial \psi_{ij}(\theta)}{\partial\theta}\right)^2 .
\]
Using
\[
\psi_{ij}(\theta)= -\frac{2\pi}{\lambda}\bm b_{ij}^{\mathsf T}\bm u(\theta),
\]
we obtain
\[
\frac{\partial \psi_{ij}(\theta)}{\partial\theta}
= -\frac{2\pi}{\lambda}\bm b_{ij}^{\mathsf T}\dot{\bm u}(\theta).
\]
Substituting this derivative into the Fisher information expression gives \eqref{eq:fisher_info}. The CRLB then follows from the standard scalar Cramér--Rao inequality.
\end{proof}

\subsection{Cycle Projection and Denoising}

Unlike minimal receivers, a multi-port correlation graph may contain cycles. We show that projecting edge phases onto the gradient (cycle-consistent) subspace reduces the expected noise energy in the locally unwrapped regime.

\begin{theorem}[Noise Reduction via Gradient Projection]
\label{thm:cycle_contraction}
Let $G$ have $C$ connected components, so $\mathrm{rank}(\mathbf{B})=N_r-C$. Let $\bm{\nu} \in \mathbb{R}^{|E|}$ be i.i.d.\ zero-mean edge noise with variance $\sigma^2$. Define
\[
\mathbf{P}\triangleq \mathbf{B}^{\mathsf T} (\mathbf{B}\mathbf{B}^{\mathsf T})^{\dagger} \mathbf{B}.
\]
Then $\mathbf{P}$ is the orthogonal projector onto $\mathrm{Im}(\mathbf{B}^{\mathsf T})$, and for $\tilde{\bm{\nu}}=\mathbf{P}\bm{\nu}$,
\begin{equation}
\mathbb{E}[\|\tilde{\bm{\nu}}\|_2^2] = (N_r-C)\sigma^2
= \frac{N_r-C}{|E|}\,\mathbb{E}[\|\bm{\nu}\|_2^2].
\label{eq:noise_reduction}
\end{equation}
\end{theorem}

\begin{proof}
By construction,
\[
\mathbf{P}\triangleq \mathbf{B}^{\mathsf T} (\mathbf{B}\mathbf{B}^{\mathsf T})^{\dagger} \mathbf{B}
\]
is the orthogonal projector onto $\mathrm{Im}(\mathbf{B}^{\mathsf T})$. Hence $\mathbf{P}$ is symmetric and idempotent, i.e.,  $
\mathbf{P}^{\mathsf T}=\mathbf{P}$ and $\mathbf{P}^2=\mathbf{P}$. Therefore, its eigenvalues are either $0$ or $1$, so $
\mathrm{Tr}(\mathbf{P})=\mathrm{rank}(\mathbf{P}).$
Since $\mathrm{Im}(\mathbf{P})=\mathrm{Im}(\mathbf{B}^{\mathsf T})$, we have
\[
\mathrm{rank}(\mathbf{P})
=\dim(\mathrm{Im}(\mathbf{B}^{\mathsf T}))
=\mathrm{rank}(\mathbf{B})
=N_r-C.
\]
Now, using $\tilde{\bm{\nu}}=\mathbf{P}\bm{\nu}$,
\[
\|\tilde{\bm{\nu}}\|_2^2
=\bm{\nu}^{\mathsf T}\mathbf{P}^{\mathsf T}\mathbf{P}\bm{\nu}
=\bm{\nu}^{\mathsf T}\mathbf{P}\bm{\nu},
\]
where the last step follows from $\mathbf{P}^{\mathsf T}=\mathbf{P}$ and $\mathbf{P}^2=\mathbf{P}$. Since $\bm{\nu}$ is i.i.d.\ zero mean with covariance $\mathbb{E}[\bm{\nu}\bm{\nu}^{\mathsf T}] = \sigma^2 \mathbf{I}_{|E|}$, the trace identity yields
\begin{equation}
\begin{aligned}
\mathbb{E}[\|\tilde{\bm{\nu}}\|_2^2]
=\mathrm{Tr}\!\big(\mathbf{P}\,\mathbb{E}[\bm{\nu}\bm{\nu}^{\mathsf T}]\big) 
=\sigma^2\mathrm{Tr}(\mathbf{P})
=(N_r-C)\sigma^2.
\end{aligned}
\end{equation}
Finally,
$\mathbb{E}[\|\bm{\nu}\|_2^2]
=\mathrm{Tr}\!\big(\mathbb{E}[\bm{\nu}\bm{\nu}^{\mathsf T}]\big)
=|E|\sigma^2$ yields
\[
\mathbb{E}[\|\tilde{\bm{\nu}}\|_2^2]
=(N_r-C)\sigma^2
=\frac{N_r-C}{|E|}\,\mathbb{E}[\|\bm{\nu}\|_2^2].
\]
\end{proof}

\subsection{Proposed Inference Pipeline}
\label{sec:pipeline}

We adopt a simple inference pipeline in which Theorem~\ref{thm:cycle_contraction} is applied before node-phase recovery. From the measured edge-phase vector $\widehat{\bm\psi}_k$, obtained by stacking the edge phases in \eqref{eq:phase_only_psi_local} according to a fixed edge orientation, we first denoise it by projection onto the gradient (cycle-consistent) subspace $
\tilde{\bm\psi}_k=\mathbf{P}\,\widehat{\bm\psi}_k.$ We then form denoised edge phasors
\begin{equation}
\tilde{\rho}_{ij}(k)\triangleq \exp(\jj\,\tilde{\psi}_{ij}(k))
\end{equation}
and estimate the latent node-phase vector $\bm z_k\in(\mathbb{S}^1)^{N_r}$, up to a global rotation, via weighted angular synchronization:
\begin{equation}
\widehat{\bm{z}}_{k}
= \argmax_{\bm{z}\in (\mathbb{S}^1)^{N_r}}
\sum_{(i,j)\in E} w_{ij}(k)\,
\Ree\!\left\{ \tilde{\rho}_{ij}(k)\, z_i z_j^* \right\}.
\label{eq:ang_sync}
\end{equation}
Communication symbols are detected by matching the recovered node-phase vector $\widehat{\bm{z}}_k$ to the receive-side spatial codebook
\begin{equation}
\mathcal{Z}_{\mathrm{rx}} \triangleq \{ \bm{z}^{(m)} \}_{m=1}^M,
\end{equation}
where $M$ denotes the modulation order. To account for the inherent global phase ambiguity, detection is performed by minimizing the Euclidean distance over all phase rotations $\alpha \in [0, 2\pi)$:
\begin{equation}
\hat{m}_k = \argmin_{m \in \{1, \dots, M\}}
\left(
\min_{\alpha \in [0, 2\pi)}
\left\| \widehat{\bm{z}}_k - e^{\jj\alpha}\bm{z}^{(m)} \right\|_2^2
\right).
\label{eq:comm_det}
\end{equation}
Once the data symbol is identified, the sensing task proceeds by removing the symbol-dependent contribution from the denoised edge phases via subtraction using the detected codeword. The resulting edge phases are then used to estimate the direction parameter $\theta$ according to the model in Sec.~\ref{sec:top-CRLB}.


\section{Simulation Results}\label{sec:result}


The G-SPMC framework is numerically validated over receiver configurations with different numbers of ports and correlation-graph topologies. Following the statistical model in Section~IV, the phase measurements are treated as locally unwrapped Gaussian variables. The practical efficacy of the proposed graph-based noise reduction is further demonstrated using a 16-symbol spatial modulation set, highlighting its benefit under higher-order spatial signaling.

Fig. \ref{fig:gspmc_comm_topology} illustrate the communication performance of G-SPMC as a function of both the receiver correlation-graph topology and the number of receive ports. Fig.~\ref{fig:gspmc_comm_topology_a} shows that, for a fixed receive-port count of \(N_r=4\), the BER decreases monotonically as the graph becomes more connected, consistent with the topologies summarized in Tab.~\ref{tab:noise_reduction}. This behavior is expected, since additional edges provide redundant phase-difference constraints that improve the reliability of the graph-based inference stage and, consequently, the codeword decision. In particular, in the moderate-to-high SNR regime, the single-cycle and complete graphs achieve a target BER at lower SNR than the path graph, highlighting a clear performance gain enabled by cycle redundancy. Fig.~\ref{fig:gspmc_comm_topology_b} further shows that, for the \textit{complete} graph topology, increasing the number of receive ports provides an additional BER improvement. This result indicates that, beyond graph connectivity, the spatial dimension of the receiver also plays a beneficial role by supplying a richer set of relative-phase observations, thereby strengthening the proposed correlation-graph-based detection.

\begin{figure}[t] \centering \subfigure[ $N_r=4$ vs. topologies]{\includegraphics[width=1.7in]{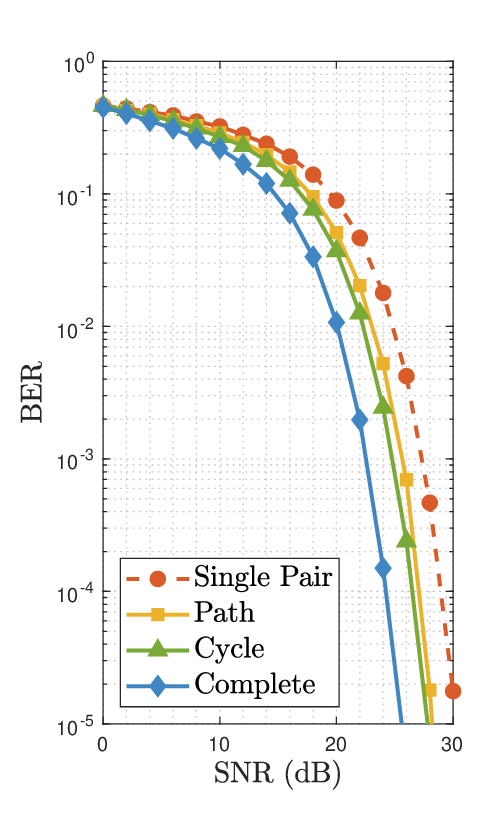}\label{fig:gspmc_comm_topology_a}} 
\subfigure[\textit{Complete} vs. port count] {\includegraphics[width=1.7in]{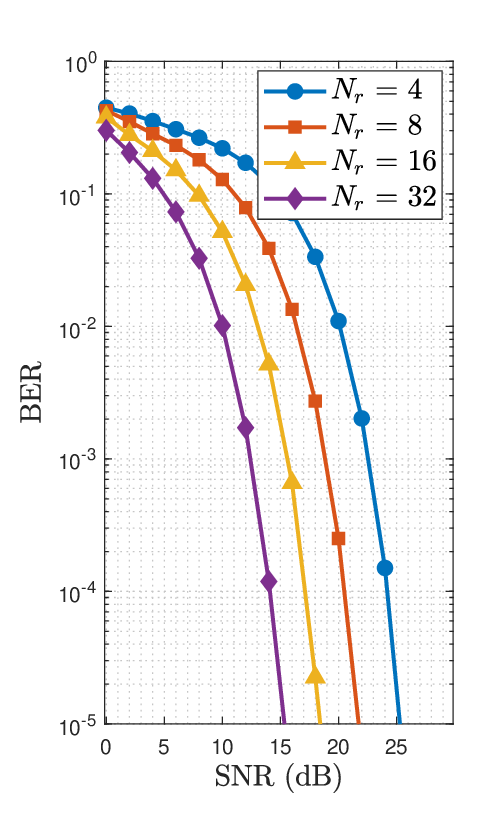}\label{fig:gspmc_comm_topology_b}} \caption{BER performance of G-SPMC for different receiver correlation-graph topologies and receive-port counts.} \label{fig:gspmc_comm_topology} 
\vspace{-0.5cm}
\end{figure}



\newcommand{\gSingle}{%
\begin{tikzpicture}[baseline=-0.5ex,scale=0.35,
  v/.style={circle,fill=black,inner sep=1.1pt},
  e/.style={line width=0.6pt}]
\node[v] (a) at (0,0) {}; \node[v] (b) at (1,0) {};
\draw[e] (a)--(b);
\end{tikzpicture}}

\newcommand{\gPath}{%
\begin{tikzpicture}[baseline=-0.5ex,scale=0.35,
  v/.style={circle,fill=black,inner sep=1.1pt},
  e/.style={line width=0.6pt}]
\node[v] (a) at (0,0) {}; \node[v] (b) at (1,0) {};
\node[v] (c) at (2,0) {}; \node[v] (d) at (3,0) {};
\node[v] (e) at (4,0) {}; \node[v] (f) at (5,0) {};
\draw[e] (a)--(b)--(c)--(d)--(e)--(f);
\end{tikzpicture}}

\newcommand{\gCycle}{%
\begin{tikzpicture}[baseline=-0.5ex,scale=0.35,
  v/.style={circle, fill=black, inner sep=1.2pt, draw=black, line width=0.3pt},
  e/.style={line width=0.6pt}]

\node[v] (v1) at (90:1.0) {};   
\node[v] (v2) at (30:1.0) {};   
\node[v] (v3) at (330:1.0) {};  
\node[v] (v4) at (270:1.0) {};  
\node[v] (v5) at (210:1.0) {};  
\node[v] (v6) at (150:1.0) {};  

\draw[e] (v1) to[bend left=20] (v2);
\draw[e] (v2) to[bend left=20] (v3);
\draw[e] (v3) to[bend left=20] (v4);
\draw[e] (v4) to[bend left=20] (v5);
\draw[e] (v5) to[bend left=20] (v6);
\draw[e] (v6) to[bend left=20] (v1);

\end{tikzpicture}}

\newcommand{\gComplete}{%
\begin{tikzpicture}[baseline=-0.5ex,scale=0.35,
  v/.style={circle, fill=black, inner sep=1.2pt, draw=black, line width=0.3pt},
  e/.style={line width=0.6pt}]
  
\node[v] (v1) at (90:1.0) {};   
\node[v] (v2) at (30:1.0) {};   
\node[v] (v3) at (330:1.0) {};  
\node[v] (v4) at (270:1.0) {};  
\node[v] (v5) at (210:1.0) {};  
\node[v] (v6) at (150:1.0) {};  

\draw[e] (v1) to[bend left=20] (v2);
\draw[e] (v2) to[bend left=20] (v3);
\draw[e] (v3) to[bend left=20] (v4);
\draw[e] (v4) to[bend left=20] (v5);
\draw[e] (v5) to[bend left=20] (v6);
\draw[e] (v6) to[bend left=20] (v1);

\draw[e] (v1) -- (v3); \draw[e] (v1) -- (v4); \draw[e] (v1) -- (v5);
\draw[e] (v2) -- (v4); \draw[e] (v2) -- (v5); \draw[e] (v2) -- (v6);
\draw[e] (v3) -- (v5); \draw[e] (v3) -- (v6);
\draw[e] (v4) -- (v6);
\end{tikzpicture}}

\begin{table}[h!]
\vspace{-0.1cm}
\centering
\caption{Noise Energy Reduction via Gradient Projection}
\label{tab:noise_reduction}
\setlength{\tabcolsep}{4pt} 
\footnotesize 
\begin{tabular}{@{}lccccc@{}} 
\toprule
\begin{tabular}{@{}c@{}}\textbf{Graph} \\ \textbf{Topology}\end{tabular}
& \textbf{
{Illustration}}
& $\mathbf{|\mathcal{E}|}$ 
& \begin{tabular}{@{}c@{}}\textbf{Indep.} \\ \textbf{Cycles}\end{tabular} 
& $\mathbf{\frac{N_r-C}{|\mathcal{E}|}}$ 
& \begin{tabular}{@{}c@{}}\textbf{Noise} \\ \textbf{Reduction}\end{tabular} \\
\midrule
Single Pair& \gSingle  & 1  & 0  & 1    & 0\%  \\
Path & \gPath        & 5  & 0  & 1    & 0\%  \\
Cycle & \gCycle       & 6  & 1  & 5/6  & 17\% \\
Complete & \gComplete    & 15 & 10 & 1/3  & 67\% \\
\bottomrule
\end{tabular}
\vspace{0cm}
\end{table}

Fig.~\ref{fig:gspmc_sensing_rmse} and Tab.~I indicate that increased connectivity also improves DoA sensing: the DoA RMSE decreases with connectivity, and the simulated curves approach the CRLB in Theorem~\ref{thm:crlb} at moderate-to-high SNR. As summarized in Tab.~I, this improvement is driven by cycle-enabled gradient projection, which suppresses the noise energy; acyclic topologies provide no reduction, whereas cyclic graphs yield measurable suppression, thereby enhancing DoA accuracy.

\begin{figure}[t]
    \centering
    \includegraphics[width=\linewidth]{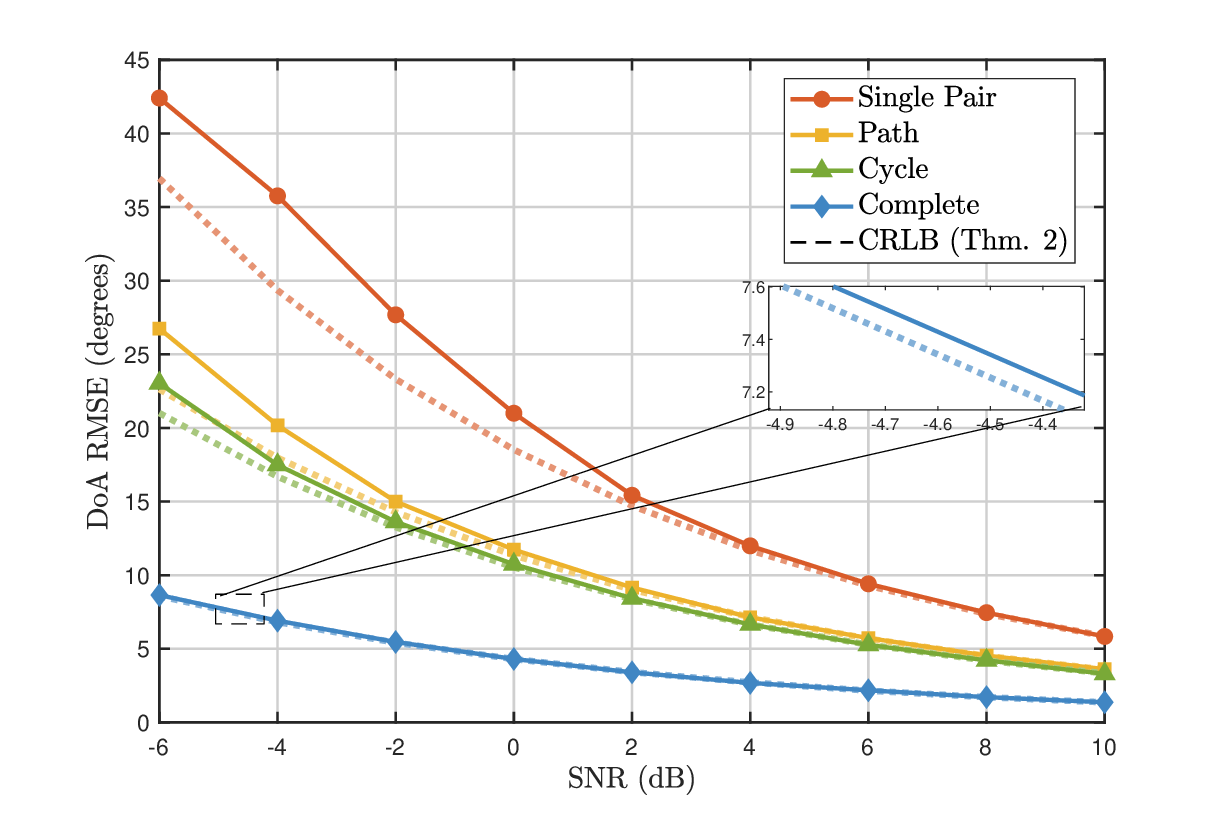}
    \caption{DoA RMSE of G-SPMC versus receiver correlation-graph topology.}
  \label{fig:gspmc_sensing_rmse}
  \vspace{-0.5cm}
\end{figure}


\section{Conclusion}
\label{sec:conclusion}
This article presented G-SPMC, a graph-based framework for JCAS with LO-free multi-port receivers. Using pairwise inter-antenna correlations, the proposed method supports both data detection and DoA estimation without carrier recovery. We showed that graph connectivity guarantees identifiability of the spatial phase state, and that projection onto the graph-consistent subspace can be explored to reduce the noise effect. Future work will extend the framework to multi-user settings.

\bibliographystyle{IEEEtran}
\bibliography{IEEEabrv, spmc_refs}

@STRING{IEEE_J_PROC       = "Proc. {IEEE}"}

@STRING{IEEE_J_MTT        = "{IEEE} Trans. Microw. Theory Techn."}

@STRING{IEEE_J_JSSC       = "{IEEE} J. Solid-State Circuits"}

@STRING{IEEE_M_COM        = "{IEEE} Commun. Mag."}

@STRING{IEEE_J_IOM = "{IEEE Internet Things Mag}"}

@ARTICLE{he2024isac,
  author={He, Dongxuan and Hou, Huazhou and Jiang, Rongkun and Yu, Xinghuo and Zhao, Zhongyuan and Mo, Yuanqiu and Huang, Yongming and Yu, Wenwu and Quek, Tony Q. S.},
  journal=IEEE_J_IOM                , 
  title={Integrating Sensing and Communication for {IoT} Systems: Task-Oriented Control Perspective}, 
  year={2024},
  volume={7},
}

@ARTICLE{yang2019hard,
  author={Yang, Xi and Matthaiou, Michail and Yang, Jie and Wen, Chao-Kai and Gao, Feifei and Jin, Shi},
  journal=IEEE_M_COM        , 
  title={Hardware-Constrained Millimeter-Wave Systems for {5G}: Challenges, Opportunities, and Solutions}, 
  year={2019},
  volume={57},
}

@INPROCEEDINGS{wentzloff2020review,
  author={Wentzloff, David D. and Alghaihab, Abdullah and Im, Jaeho},
  booktitle = {Proc. IEEE CICC}, 
  title={Ultra-Low Power Receivers for IoT Applications: A Review}, 
  year={2020},
 }

@ARTICLE{wurx2022,
  author={Mercier, Patrick P. and Calhoun, Benton H. and Wang, Po-Han Peter and Dissanayake, Anjana and Zhang, Linsheng and Hall, Drew A. and Bowers, Steven M.},
  journal=IEEE_J_JSSC       , 
  title={Low-Power {RF} Wake-Up Receivers: Analysis, Tradeoffs, and Design}, 
  year={2022},
  volume={2},
  pages={144-164},

}

@STRING{J_LPEA = "{J. Low Power Electron. Appl.}"}

@Article{jlpea2025wvrx,
AUTHOR = {Chen, Suhao and Yu, Xiaopeng and Huang, Xiongchun},
TITLE = {Wake-Up Receivers: A Review of Architectures Analysis, Design Techniques, Theories and Frontiers},
JOURNAL = J_LPEA,
VOLUME = {15},
YEAR = {2025},
}

@INPROCEEDINGS{selfdemod2010,
  author  = {X. Li and P. Baltus and P. van Zeijl and D. Milosevic and A. van Roermund},
  booktitle = {Proc. IEEE CICC}, 
  title   = {A 70-{GHz} 10.2-{mW} Self-Demodulator for {OOK} Modulation in 65-nm {CMOS} Technology},
  year    = {2010}
 }

@ARTICLE{godara97,
  author={Godara, L.C.},
  journal=IEEE_J_PROC       , 
  title={Application of antenna arrays to mobile communications. {II}. Beam-forming and direction-of-arrival considerations}, 
  year={1997},
  volume={85},
  pages={1195-1245},


}

@article{rajabi2018irr,
  author  = {M. Rajabi and N. Pan and S. Claessens and S. Pollin and D. Schreurs},
  title   = {Modulation Techniques for Simultaneous Wireless Information and Power Transfer With an Integrated Rectifier-Receiver},
  journal = IEEE_J_MTT,
  year={2018},
  volume={66},
  pages={2373-2385},

}

@book{fisher1993circular, place={Cambridge}, title={Statistical Analysis of Circular Data}, publisher={Cambridge University Press}, author={Fisher, N. I.}, year={1993}}

\end{document}